\documentclass[a4paper]{scrartcl}
\usepackage{amsmath, amsthm, amssymb}
\usepackage{amsmath, amssymb}
\usepackage{enumitem}
\usepackage{ucs}
\usepackage[utf8x]{inputenc}
\usepackage{pgf}
\usepackage{tikz}
\usetikzlibrary{arrows,automata}
\usetikzlibrary{positioning}

\newtheorem{thm}{Theorem}[section]
\newtheorem{cor}[thm]{Corollary}
\newtheorem{lem}[thm]{Lemma}

\newtheorem{prop}[thm]{Proposition}
\newtheorem{clm}[thm]{Claim}

\theoremstyle{remark}

\theoremstyle{definition}
\newtheorem{defi}[thm]{Definition}

\newcommand{\vp}{\ensuremath{\mathsf{VP}}}

\newcommand{\vpe}{\ensuremath{\mathsf{VP}_e}}
\newcommand{\vpws}{\ensuremath{\mathsf{VP}_{ws}}}
\newcommand{\vnp}{\ensuremath{\mathsf{VNP}}}

\newcommand{\lc}{\ensuremath{\mathsf{LOGCFL}}}

\newcommand{\push}{\mathop{push}}
\newcommand{\pop}{\mathop{pop}}
\newcommand{\nop}{\mathop{nop}}
\newcommand{\writ}{\mathop{write}}
\newcommand{\delete}{\mathop{delete}}
\newcommand{\op}{\mathop{op}}

\title{Arithmetic Branching Programs with Memory}
\author{Stefan Mengel\thanks{Partially supported by DFG grants BU 1371/2-2 and BU 1371/3-1.}\\Institute of Mathematics\\ University of Paderborn\\ D-33098 Paderborn, Germany\\ {\small\texttt{smengel@mail.uni-paderborn.de}} }

\begin{document}

\maketitle

\begin{abstract}
We extend the well known characterization of $\vpws$ as the class of polynomials computed by polynomial size arithmetic branching programs to other complexity classes. In order to do so we add additional memory to the computation of branching programs to make them more expressive. We show that allowing different types of memory in branching programs increases the computational power even for constant width programs. In particular, this leads to very natural and robust characterizations of $\vp$ and $\vnp$ by branching programs with memory.
\end{abstract}

\section{Introduction}

Arithmetic Branching Programs (ABPs) are a well studied model of computation in algebraic complexity: They were already used by Valiant in the \vnp-completeness proof of the permanent \cite{Valiant1979} and have since then contributed to the understanding of arithmetic circuit complexity (see e.g. \cite{ni91,Koiran12}). The computational power of ABPs is well understood: They are equivalent to both skew and weakly skew arithmetic circuits and thus capture the determinant, matrix power and other natural problems from linear algebra~\cite{MP08}. The complexity of bounded width ABPs is also well understood: In a parallel to Barrington's Theorem \cite{bar89}, Ben-Or and Cleve \cite{BC92} proved that polynomial size ABPs of bounded width are equivalent to arithmetic formulas.

We modify ABPs by giving them memory during their computations and ask how this changes their computational power. There are several different motivations for doing this: We define branching programs with stacks, that are an adaption of the nondeterministic auxiliary pushdown automaton (NAuxPDA) model to the arithmetic circuit model. The NAuxPDA-characterization of \lc\;has been very successful in the study of this class and has contributed a lot to its understanding. We give a characterization of \vp\;-- a class that is well known for its apparent lack of natural characterizations. In the Boolean setting graph connectivity problems on edge-labeled graphs that are similar to our ABPs with stacks have been shown to be complete for $\lc$ \cite{SkyumV85,WeberS07}. One motivation for adapting these results to the arithmetic circuit setting is the hope that one can apply techniques from the NAuxPDA setting to arithmetic circuits. We show that this is indeed applicable by presenting an adaption of a proof of Niedermeier and Rossmanith \cite{NR95} to give a straightforward proof of the classical parallelization theorem for \vp\;first proved by Valiant et al.\ \cite{VSBR83}. 

Another motivation is that our modified branching programs in different settings give various very similar characterizations of different arithmetic circuit classes. This allows us to give a new perspective on problems like \vp\;vs.\ \vpws, \vp\;vs.\ \vnp\; that are classical question from arithmetic circuit complexity. %This question was already considered by the author in \cite{men11}.
This is similar to the motivation that Kintali~\cite{Kintali10} has for studying similar graph connectivity problems in the Boolean setting.

Finally, all modifications we make to ABPs are straightforward and natural. The basic question is the following: ABPs are in a certain sense a memoryless model of computation. At each point of time during the computation we do not have any information about the history of the computation sofar apart from the state we are in. So what happens if we allow memory during the computation? Intuitively, the computational power should increase, and we will see that it indeed does (under standard complexity assumptions of course). How do different types of memory compare? What is the role of the width of the branching programs if we allow memory? In the remainder we will answer several of these questions. 

The structure of the paper is a follows: After some preliminaries we start off with ABPs that may use a stack during their computation. We show that they characterize~\vp, consider several restrictions and give a proof of the parallelization theorem for~\vp. Next we consider ABPs with random access memory, show that they characterize~\vnp\;and consider some restrictions of them, too. %We close the paper with some open question.

\section{Preliminaries}

\subsection{Arithmetic circuits}

We briefly recall the relevant definitions from arithmetic circuit complexity. A more thorough introduction into arithmetic circuit classes can be found in the book by Bürgisser~\cite{Bur00}. Newer insights into the nature of~\vp\;and especially~\vpws\;are presented in the excellent paper of Malod and Portier~\cite{MP08}.

An {\em arithmetic circuit} over a field $\mathbb{F}$ is a labeled directed acyclic graph (DAG) consisting of vertices or gates with indegree or fanin $0$ or $2$. The gates with fanin $0$ are called input gates and are labeled with constants from $\mathbb{F}$ or variables $X_1, X_2, \ldots$. The gates with fanin $2$ are called computation gates and are labeled with $\times$ or $+$. 

The polynomial computed by an arithmetic circuit is defined in the obvious way: An input gates computes the value of its label, a computation gate computes the product or the sum of its childrens' values, respectively. We assume that a circuit has only one sink which we call the output gate. We say that the polynomial computed by the circuit is the polynomial computed by the output gate. The \emph{size} of an arithmetic circuit is the number of gates. The \emph{depth} of a circuit is the length of the longest path from an input gate to the output gate in the circuit. 

We also consider circuits in which the $+$-gates may have unbounded fanin. We call these circuits {\em semi-unbounded circuits}. Observe that in semi-unbounded circuits $\times$-gates still have fanin $2$. A circuit is called \emph{multiplicatively disjoint} if for each $\times$-gate $v$ the subcircuits that have the children of $v$ as output-gates are disjoint. A circuit is called \emph{skew}, if for all of its $\times$-gates one of the children is an input gate.

We call a sequence $(f_n)$ of multivariate polynomials a family of polynomials or \emph{polynomial family}. We say that a polynomial family is of polynomial degree, if there is a univariate polynomial $p$ such that $\deg(f_n) \le p(n)$ for each $n$.
\vp \;is the class of polynomial families of polynomial degree computed by families of polynomial size arithmetic circuits. We will use the following well known characterizations of~\vp:

\begin{thm}(\cite{VSBR83,MP08})\label{thm:vpknown}
 Let $(f_n)$ be a family of polynomials. The following statements are equivalent:
\begin{enumerate}
 \item $(f_n) \in \vp$
 \item $(f_n)$ is computed by a family of multiplicatively disjoint polynomial size circuits.
 \item $(f_n)$ is computed by a family of semi-unbounded circuits of logarithmic depth and polynomial size.
\end{enumerate}
\end{thm}

\vpe\;is defined analogously to \vp\;with the circuits restricted to trees. By a classical result of Brent \cite{Br76}, \vpe\;equals the class of polynomial families computed by arithmetic circuits of depth $O(\log(n))$. \vpws~is the class of families of polynomials computed by families of skew circuits of polynomial size. Finally, a family $(f_n)$ of polynomials is in \vnp, if there is a family $(g_n)\in \vp$ and a polynomial $p$ such that $f_n(X) = \sum_{e\in \{0,1\}^{p(n)}} g_n(e,X)$ for all $n$ where $X$ denotes the vector $(X_1, \ldots, X_{q(n)})$ for some polynomial~$q$.

A polynomial $f$ is called a \emph{projection} of $g$ (symbol: $f \le g$), if there are values $a_i \in \mathbb{F} \cup \{X_1, X_2, \ldots\}$ such that $f(X) = g(a_1, \ldots, a_q)$. A family $(f_n)$ of polynomials is a $p$-projection of $(g_n)$ (symbol: $(f_n) \le_p (g_n)$), if there is a polynomial $r$ such that $f_n \le g_{r(n)}$ for all $n$. As usual we say that $(g_n)$ is hard for an arithmetic circuit class $\mathcal{C}$ if for every $(f_n) \in \mathcal{C}$ we have $(f_n) \le_p (g_n)$. If further $(g_n) \in \mathcal{C}$ we say that $(g_n)$ is $\mathcal{C}$-complete.

The following criterion by Valiant \cite{Valiant1979} for containment in $\vnp$ is often helpful:

\begin{lem}[Valiant's criterion]\label{lem:criterion}
 Let $\phi:\{0,1\}^*\rightarrow \mathbb{N}$ be a function in $\mathsf{\#P/poly}$, Then the family $(f_n)$ of polynomials defined by \[f_n = \sum_{e\in \{0,1\}^n} \phi(e) \prod_{i=1}^n X_i^{e_i}\] is in $\vnp$.
\end{lem}

\subsection{Arithmetic branching programs}

The second common model of computation in arithmetic circuit complexity are arithmetic branching programs.

\begin{defi}
 An \emph{arithmetic branching program} (ABP) $G$ is a DAG with two vertices $s$ and $t$ and an edge labeling $w : E \rightarrow \mathbb{F} \cup \{X_1, X_2, \ldots\}$. A path $P = v_1 v_2\ldots v_r$ in $G$ has the \emph{weight} $w(P) := \prod_{i=1}^{r-1} w(v_iv_{i+1})$. Let $v$ and $u$ be two vertices in $G$, then we define \[f_{v,u} = \sum_{P} w(P),\] where the sum is over all $v$-$u$-paths $P$. The ABP $G$ computes the polynomial $f_G = f_{s,t}$. The \emph{size} of $G$ is the number of vertices of $G$.
\end{defi}

Malod and Portier proved the following theorem:

\begin{thm}(\cite{MP08})
 $(f_n)\in \vpws$, iff $(f_n)$ is computed by a family of polynomial size~ABPs.
\end{thm}

\begin{defi}
 An ABP of width $k$ is an ABP in which all vertices are organized into layers $L_i, i\in \mathbb{N}$, there are only edges from layer $L_i$ to $L_{i+1}$ and the number of vertices in each layer $L_i$ is at most $k$.
\end{defi}

The computational power of ABPs of constant width was settled by Ben-Or and~Cleve:

\begin{thm}(\cite{BC92})
 $(f_n)\in \vpe$, iff $(f_n)$ is computed by a family of polynomial size ABPs of constant width.
\end{thm}

\section{Stack branching programs}

\subsection{Definition}

Let $S$ be a set called \emph{symbol set}. For a symbol $s\in S$ we define two \emph{stack operations}: $\push(s)$ and $\pop(s)$. Additionally we define the stack operation $\nop$ without any arguments. A \emph{sequence of stack operations} on $S$ is a sequence $\mathop{op}_1 \mathop{op}_2 \ldots \mathop{op}_r$, where either $\op_i = \bar{\op}_i(s_i)$ for $\bar{\op}_i \in \{\push , \pop\}$ and $s_i\in S$ or $\op_i= \nop$. \emph{Realizable sequences} of stack operations are defined inductively:

\begin{itemize}
 \item The empty sequence is realizable.
 \item If $P$ is a realizable sequence of stack operations, then $\push(s) P \pop(s)$ is realizable for all $s\in S$. Also $\nop P$ and $P\nop$ are realizable sequences.
 \item If $P$ and $Q$ are realizable sequences of stack operations, then $PQ$ is a realizable sequence.
\end{itemize}

\begin{defi}
A \emph{stack branching program} (SBP) $G$ is an ABP with an additional edge labeling $\sigma: E \rightarrow \{ \op(s) \mid \op \in \{\push , \pop\}, s \in S\}\cup \{\nop\}$. A path $P = v_1 v_2 \ldots v_r$ in $G$ has the sequence of stack operations $\sigma(P) := \sigma(v_1v_2)\sigma(v_2v_3) \ldots \sigma(v_{r-1}v_r)$. If $\sigma(P)$ is realizable we call $P$ a \emph{stack-realizable path}. The SBP $G$ computes the polynomial \[f_G = \sum_{P} w(P),\] where the sum is over all stack-realizable $s$-$t$-paths $P$.
\end{defi}

It is helpful to interpret the stack operations as operations on a real stack that happen along a path through $G$. On an edge $uv$ with the stack operation $\sigma(uv) = \push(s)$ we simply push $s$ onto the stack. If $uv$ has the stack operation $\sigma(uv) = \pop(s)$ we pop the top symbol of the stack. If it is $s$ we continue the path, but if it is different from $s$ the path is not stack realizable and we abort it. $\nop$ stands for ``no operation'' and thus as this name suggests the stack is not changed on edges labelled with $nop$. Realizable paths are exactly the paths on which we can go from $s$ to $t$ in this way without aborting while starting and ending with an empty stack.

To ease notation we sometimes call edges $e$ with $\sigma(e)= \push(s)$ for an $s\in S$ simply $\push$-edges. $\pop$-edges and $\nop$-edges are defined in the obvious analogous way.

It will sometimes be convenient to consider only SBPs that have no $\nop$-edges. The following easy proposition shows that this is not a restriction.

\begin{prop}\label{prop:nonopSBP}
 Let $G$ be an SBP of size $s$. There is an SBP $G'$ of size $O(s^2)$ such that $f_G= f_{G'}$ and $G'$ does not contain any $\nop$-edges. If $G$ is layered with width $k$, then $G'$ is layered, too, and has width at most $k^2$.
\end{prop}
\begin{proof}
 The idea of the construction is to subdivide every edge of $G$. So let $G$ be an SBP with vertex set $V$ and edge set $E$. Let $\sigma$ and $w$ be the stack symbol labeling and the weight function, respectively. $G'$ will have the vertex set $V\cup\{v_e\mid e\in E\}$, stack symbol labeling $\sigma'$ and weight function $w'$. The construction goes as follows: For each edge $e=uv\in E$ the SBP $G'$ has the edges $uv_e, v_ev$. We set $w'(uv_e) := w(uv)$ and $w'(v_ev):=1$. If $e$ is a $\nop$-edge we set $\sigma'(uv_e):= \push(s)$ and $\sigma'(v_ev)= \pop(s)$ for an arbitrary stack symbol $s$. Otherwise, both $uv_e$ and $v_ev$ get the stack operation $\sigma(uv)$.
 
 It is easy to verify that $G'$ has all desired properties.
\end{proof}

\subsection{Characterizing VP}%\vp}

In this section we show that stack branching programs of polynomial size characterize~\vp.

\begin{thm}\label{thm:charvp}
$(f_n)\in \vp$, iff $(f_n)$ is computed by a family of polynomial size~SBPs.
\end{thm}

We the two direction of Theorem \ref{thm:charvp} independently.

\begin{lem}
 If $(f_n)$ is computed by a family of polynomial size SBPs, then $(f_n) \in \vp$.
\end{lem}
\begin{proof}
 Let $(G_n)$ be a family of SBPs computing $(F_n)$, of size at most $p(n)$ for a polynomial~$p$. Observe that $\deg(G_n) \le p(n)$, so we only have to show that we can compute the $G_n$ by polynomial size circuits $C_n$.

Let $G=G_n$ be an SBP with $m$ vertices, source $s$ and sink $t$.
The construction of $C=C_n$ uses the following basic observation: Every stack-realizable path $P$ of length $i$ between two vertices $v$ and $u$ can be uniquely decomposed in the following way. There are vertices $a,b, c \in V(G)$ and a symbol $s\in S$ such that there are edges $va$ and $bc$ with $\sigma(va) = \push(s)$ and $\sigma(bc)= \pop(s)$. Furthermore there are stack-realizable paths $P_{ab}$ from $a$ to $b$ and $P_{cu}$ from $c$ to $u$ such that $length(P_{ab}) + length(P_{cu}) = i - 2$ and $P=vaP_{ab}bcP_{cu}$. The paths $P_{ab}$ and $P_{cu}$ may be empty. We define $w(u,v,i):= \sum_P w(P)$ where the sum is over all stack-realizable $s$-$t$-paths of length $i$.

The values $w(v,u, i)$ can be computed efficiently with a straightforward dynamic programming approach. First observe that $w(v,u,i) = 0$ for odd $i$. For $i=0$ we set $w(v,u,0) = 0$ for $v\neq u$ and $w(v,v,0)= 1$. For even $i > 0$ we get \[ w(v,u,i) = \sum_{a,b,c, j, s} w(v,a) w(a,b,j) w(b,c) w(c,u, i-j-2), \] where the sum is over all $s\in S$, all $j \le i-2$ and all $a,b,c$ such that $\sigma(va) = \push(s)$ and $\sigma(bc)= \pop(s)$. With this recursion formula we can compute alNote that Kintali proved a similar result for the Turing machine setting.

l $w(v,u,i)$ with a polynomial number of arithmetic operations. Having computed all $w(v,u,i)$ we get $f_G = \sum_{i \in [m]} w(s,t,i)$.
\end{proof}

The more involved direction of the proof of Theorem \ref{thm:charvp} will be the second direction. To prove it it will be convenient to slightly relax our model of computation. A \emph{relaxed SBP}~$G$ is an SBP where the underlying directed graph is not necessarily acyclic. To make use of cyclicity we do not consider paths in a relaxed SBP $G$ but \emph{walks}, i.e.\ vertices and edges of $G$ may be visited several times. \emph{Realizable walks} are defined completely analogously to realizable paths. Also the weight $w(P)$ of a walk is defined in the obvious way. Clearly, we cannot define the polynomial computed by a relaxed ABP by summing over the weight of all realizable walks, because there may be infinitely many of them and they may be arbitrarily long. Hence, we define for each pair $u,w$ of vertices and for each integer $m$ the polynomial 
\[f_{u,v,m}:= \sum_{P} w(P),\] where the sum is over all stack-realizable $u$-$v$-walks $P$ in $G$ that have length $m$. Furthermore, we say that for each $m$ the relaxed SBP $G$ computes the polynomial $f_{G,m}:=f_{s,t,m}$.

The connection to SBPs is given by the following straight-forward lemma.

\begin{lem}\label{lem:relaxed}
 Let $G$ be a relaxed SBP and $m\in \mathbb{N}$. Then for each $m$ there is an SBP $G'_m$ of size $m|G|$ that computes $f_{G,m}$.
\end{lem}
\begin{proof}
 The idea is to unwind the computation of the relaxed SBP into $m$ layers. Let $G=(V,E,w,\sigma)$, then for each $v\in V$ the SBP $G'$ has $m$ copies $\{v_1, \ldots, v_m\}$. For each $uv\in E$ the SBP $G'$ had the edges $u_iv_{i+1}$ for $i\in [m-1]$ with weight $w(u_iv_{i+1}) := w(uv)$ and stack operation $\sigma(u_iv_{i+1}) := \sigma(uv)$. This completes the construction of $G'$. 
 
 Clearly, $G'$ indeed computes $f_{G,m}$ and has size $m|G|$.
\end{proof}

To prove the characterization of $\vp$ we show the following rather technical proposition:

\begin{prop}\label{prop:technical}
 Let $C$ be a multiplicatively disjoint arithmetic circuit. For each $v\in V$ we denote by $C_v$ the subcircuit of $C$ with output $v$ and we denote by $f_v$ the polynomial computed by $C_v$. Then there is a relaxed SBP $G=(V,E,w, \sigma) $ of size at most $2|C|(|C|+1) + 3(|C|)$ such that for each $v\in V$ there is a pair $v_-, v_+ \in V$ and an integer $m_v\le 4 |C_v|$ with 
 \begin{itemize}
  \item $f_v= f_{v_-, v_+, m_v}$, and 
  \item there is no stack-realizable walk from $v_-$ to $v_+$ in $G$ that is shorter than $m_v$.
 \end{itemize}
\end{prop}
\begin{proof}
 We construct $G$ iteratively along a topological order of $C$ by adding new vertices and edges, starting from the empty relaxed SBP.
 
 Let first $v$ be an input of $C$ with label $X$. We add two new vertices $v_-, v_+$ to $G$ and the edge $v_-v_+$ with weigth $w(v_-v_+) = X$ and stack-operation $\sigma(v_-v_+):=\nop$. Furthermore, $m_v:=1$. Clearly, none of the polynomials computed before change and the size of the relaxed SBP grows only by $2$. Thus all statements of the proposition of fulfilled.
 
 Let now $v$ be an addition gate with children $u,w$. By induction $G$ contains vertices $u_-, u_+, w_-, w_+$ and there are $m_u, m_v$ such that $f_{u_-,u_+, m_u}=f_u$ and $f_{w_-,w_+, m_w}=f_w$. Assume w.l.o.g.\ $m_u \ge m_w$. We add two new vertices $v_-, v_+$ to $G$. Furthermore, we add a directed path of length $m_u-m_w$ with start vertex $v_s$ and end vertex $v_t$ to $G$. We add the edges $v_-u_-$, $v_-v_s$, $v_tw_-$, $u_+v_+$ and $w_+v_+$. All edges we add get weight $1$. Furthermore, we set $\sigma(v_-u_-) := \push(vu)$, $\sigma(u_+v_+) := \pop(vu)$, $\sigma(v_-v_s) := \push(vw)$ and $\sigma(w_+v_+) := \pop(vw)$ for new stack symbols $vu$ and $vw$. All other edges we added are $\nop$-edges. Finally, set $m_v:=m_u+2$.
 
 Let us first check that $G$ computes the correct polynomials. First observe that the edges we added do not allow any new walks between old vertices, so we still compute all old polynomials by induction. Thus we only have to consider the realizable $v_-$-$v_+$-walks of length~$m_v$. Each of these either starts with the edge $v_-u_-$ or the edge $v_-v_s$. In the first case, because of the stack symbols the walk must end with the edge $u_+v_+$. Thus the realizable $v_-v_+$-walks of length $m_v$ that start with $v_-u_-$ contribute exactly the same weight as the realizable $u_-$-$u_+$-walks of length $m_u$ which is exactly $f_u$ by induction. Moreover, every $v_-v_+$-walks of length $m_v$ that start with $v_-v_s$ first makes $m_u-m_w$ unweighted steps to $w_-$ and ends with the edge $w_+v_+$. Thus, these walks contribute exactly the same as the stackrealizable $w_-$-$w_+$ walks of length $m_v-2-(m_u-m_w) = m_w$, so they contribute $f_w$. Combining all walks we get $f_{v_-, v_+, m_v}= f_u+f_w = f_v$ as desired. 
 
 We have $m_v = m_u + 2 \le 4 |C_u|+2 \le 4 |C_v|$ where the first inequality is by induction and the second inequality follows from the fact that $v$ is not contained in $C_u$ and thus $|C_v|> |C_u|$. To see the bound on $|G|$ let $s$ be the size of $G$ before adding the new edges and vertices. By induction $s\le 2(|C_v|-1)(|C_v|-1+1) + 3(|C_v|-1)$. We have added $2+m_u-m_v +1$ vertices and thus $G$ has now size $s + 3 +m_u-m_v\le s+3+m_u$. But we have $m_u\le 4 |C_u|\le 4|C_v|$ and thus the number of vertices in $G$ is at most $2(|C_v|-1)|C_v| + 3(|C_v|-1) + 3 + 4|C_v| \le 2|C|(|C|+1)+ 3(|C_v|)$. This completes the case that $v$ is an addition gate.
 
 Let now $v$ be a multiplication gate with children $u,w$. As before, $G$ already contains $u_-, u_+, w_-, w_+$ and there are $m_u, m_v$ with the desired properties. We add three vertices $v_-$, $v_+$ and $v_i$ and the edges $v_-u_-$, $u_+v_i$, $v_iw_-$ and $w_+v_+$ all with weight~$1$. The new edges have the stack symbols $\sigma(v_-u_-) := \push(vu)$, $\sigma(u_+v_i) := \pop(vu)$, $\sigma(v_iw_-) := \push(vw)$ and $\sigma(w_+v_+) := \pop(vw)$ for new stack symbols $vu$ and $vw$. Finally, set $m_v := m_u + m_w + 4$.
 
 Clearly, no stack-realizable walk between any pair of old vertices can traverse $v_-$, $v_+$ or $v_i$ and thus these walks still compute the same polynomials as before. Thus we only have to analyse the $v_-$-$v_+$-walks of length $m_v$ in $G$. Let $P$ be such a walk. Because of the stack symbols $vu$ and $vw$ the walk $P$ must have the structure $P=v_- u_-P_1u_+ v_i w_- P_2 w_+ v_+$ where $P_1$ and $P_2$ are a stack-realizable $u_-$-$u_+$-walk and a stack-realizable $w_-$-$w_+$-walk, respectively. The walk $P$ is of length $m_v$ and thus $P_1$ and $P_2$ must have the combined length $m_u+m_w$. But by induction $P_1$ must at least have length $m_u$ and $P_2$ must have at least length $m_w$, so it follows that $P_1$ has length exactly $m_u$ and $P_2$ has length exactly $m_w$. The walks $P_1$ and $P_2$ are independent and thus we have $f_{v_-,v_+, m_v}= f_{u_-,u_+, m_u}f_{w_-,w_+, m_w} = f_u f_w$ as desired.
 
 The circuit $C$ is multiplicatively disjoint and thus we have $|C_v|= |C_u| + |C_w|+1$. It follows that $m_v = m_u + m_w + 4 \le 4|C_u|+4|C_w|+4 = 4|C_v|$ where we get the inequality by induction. The relaxed SBP grows only by $3$ vertices which gives the bound on the size of $G$. This completes the proof for the case that $v$ is an addition gate and hence the proof of the lemma.
\end{proof}

Now the second direction of Theorem \ref{thm:charvp} is straight-forward.

\begin{lem}\label{lem:charvpseconddirection}
 Every family $(f_n) \in \vp$ can be computed by a family of SBPs of polynomial size.
\end{lem}
\begin{proof}
Given a family $(C_n)$ of multiplicatively disjoint arithmetic circuits of polynomial size, first turn them into relaxed SBPs of polynomial size and polynomial $m$ with Proposition \ref{prop:technical} and then turn those relaxed SBPs into SBPs with Lemma \ref{lem:relaxed}. It is easy to check that the resulting SBPs have polynomial size.
\end{proof}

\subsection{Stack branching programs with one stack symbol}

It is easy to see, that the number of symbols used in SBPs can be lowered to $2$ without loss of computational power and with only logarithmic overhead in the size (see also Section \ref{sct:widthreductionSBP}. Therefore the only meaningful restriction of the size of the symbol set is the restriction to a set only consisting of one single symbol. The following fairly straightforward lemma shows that doing so indeed decreases the computational power.
Note that Kintali proved a similar result for the Turing machine setting.

\begin{lem}
 $(f_n)\in \vpws$ if and only if it can be computed by polynomial size SBPs with one stack symbol.
\end{lem}
\begin{proof}
 The direction from left to right is easy: Simply interpret each edge $e$ of an ABP $G$ as a $\nop$-edge.

For the other direction the key insight is that if one has only one stack symbol one only has to keep track of the size of the stack at any point in the path. But this height can be encoded by vertices of an ABP. So let $G$ be a SBP of size $m$. It is clear that the stack cannot be higher than $m$ on any path through $G$. We construct an ABP $G'$ that has for every vertex $v$ in $G$ the $m+1$ vertices $v_0, v_1, \ldots v_m$. If $vu$ is a $\push$-edge in $G$, we connect $v_i$ to $u_{i+1}$ for $i = 0, \ldots , m-1$ in $G'$. If $vu$ is a $\pop$-edge in $G$, we add $v_i u_{i-1}$ for $i=1, \ldots, m$ to $G'$. All these edges get the same weight as $vu$ in the $G$. It is easy to see that every stack-realizable path $P$ in the SBP $G$ corresponds directly to a path $P'$ in the ABP $G'$ and $P$ and $P'$ have the same weight. Thus $G$ and $G'$ compute the same polynomial. Moreover, $|G'| = (m+1) |G|$ which completes the proof.
\end{proof}

\subsection{Width reduction}\label{sct:widthreductionSBP}

In this section we show that unlike for ordinary ABPs bounding the width of SBPs does not decrease the computational power: Polynomial size SBPs with at least 2 stack symbols and width $2$ can still compute every family in~\vp.

\begin{lem}
 Every family $(f_n) \in \vp$ can be computed by a SBP of width $2$ with the stack symbol set $\{0,1\}$.
\end{lem}
\begin{proof}
 The idea of the proof is to start from the characterization of $\vp$ by SBPs from Theorem \ref{thm:charvp}. We use the stack to remember which edge will be used next on a realizable path through the branching program. We will show how this can be done with width $2$ SBPs with a bigger stack symbol size. In a second step we will seee how to reduce the stack symbol set to $\{0,1\}$.
 
 So let $(G_n)$ be a family of SBPs. Fix $n$ and let $G:=G_n$ with vertex set $V$ and edge set $E$. Furthermore, let $w$ be the weight function, $\sigma$ the stack operation labeling and $S$ the stack symbol of $G$. Let $s$ and $t$ be the source and the sink of the SBP $G$. We assume without loss of generality that $s$ has one single outgoing edge $e_s$. Furthermore $t$ is only entered by one $\nop$-edge $e_t$ with weight $1$. We will construct a new SBP $G'$ with weight function $w'$ and stack operation labeling $\sigma'$. $G'$ will have stack symbol set $S\cup E$. For each edge $e$ with a successor edge $e'$ the SBP$G'$ contains a gadget $G_{e,e'}$. The vertex set of $G_{e,e'}$ is $\{v_{e,e'}^1, v_{e,e'}^2, v_{e,e'}^3, v_{e,e'}^4, v_{e,e'}^5, v_{e,e'}^6\}$. These vertices are connected to a DAG by the edges $\{v_{e,e'}^1v_{e,e'}^2, v_{e,e'}^1v_{e,e'}^3, v_{e,e'}^2v_{e,e'}^4, v_{e,e'}^3v_{e,e'}^5, v_{e,e'}^4v_{e,e'}^6, v_{e,e'}^5v_{e,e'}^6\}$. All these edges have weight~$1$ except for $v_{e,e'}^2v_{e,e'}^4$ for which we set $w'(v_{e,e'}^2v_{e,e'}^4):=w(e)$. We call $v_{e,e'}^2v_{e,e'}^4$ the \emph{weighted edge} of $G_{e,e'}$. Furthermore we set $\sigma(v_{e,e'}^1v_{e,e'}^2):=\pop(e)$, $\sigma(v_{e,e'}^2v_{e,e'}^4):=\sigma(e)$, $\sigma(v_{e,e'}^4v_{e,e'}^6):=\push(e')$. All other edges are $\nop$-edges. The construction of $G_{e,e'}$ is illustrated in Figure \ref{fig:gadget}.
 
 \makeatletter
 \tikzoption{above of}{\tikz@of{#1}{90}}%
\tikzoption{below of}{\tikz@of{#1}{-90}}%
\tikzoption{left of}{\tikz@of{#1}{180}}%
\tikzoption{right of}{\tikz@of{#1}{0}}%
\tikzoption{above left of}{\tikz@of{#1}{135}}%
\tikzoption{below left of}{\tikz@of{#1}{-135}}%
\tikzoption{above right of}{\tikz@of{#1}{45}}%
\tikzoption{below right of}{\tikz@of{#1}{-45}}%
\def\tikz@of#1#2{%
  \def\tikz@anchor{center}%
  \let\tikz@do@auto@anchor=\relax%
  \tikz@addtransform{%
    \expandafter\tikz@extract@node@dist\tikz@node@distance and\pgf@stop%
    \pgftransformshift{\pgfpointpolar{#2}{\tikz@extracted@node@distance}}}%
  \def\tikz@node@at{\pgfpointanchor{#1}{center}}}
\def\tikz@extract@node@dist#1and#2\pgf@stop{%
  \def\tikz@extracted@node@distance{#1}}
  \makeatother

\begin{figure}[t]
 \begin{center}
 \begin{tikzpicture}[->,>=stealth',shorten >=1pt,auto,node distance=2.8cm,
                    semithick]
%   \tikzstyle{every state}=[fill=red,draw=none,text=white]

  \node[state]         (A)                    {$v_{e,e'}^1$};
  \node[state]         (D) [right of=A] {$v_{e,e'}^3$};
  \node[state]         (B) [above of=D] {$v_{e,e'}^2$};
  \node[state]         (C) [right of=B] {$v_{e,e'}^4$};
  \node[state]         (E) [right of=D]       {$v_{e,e'}^5$};
  \node[state]         (F) [right of=E]       {$v_{e,e'}^6$};

  \path (A) edge              node {$\pop(e)$} (B)
            edge              node {} (D)
        (B) edge              node {$w(e)/\sigma(e)$} (C)
        (C) edge              node {$\push(e')$} (F)
        (D) edge              node {} (E)
        (E) edge              node {} (F);
\end{tikzpicture}
 \caption{The gadget $G_{e,e'}$. We illustrate only the weight of the weighted edges. All edges without stack operation label are $\nop$-edges.}\label{fig:gadget}
 \end{center}
 \end{figure}
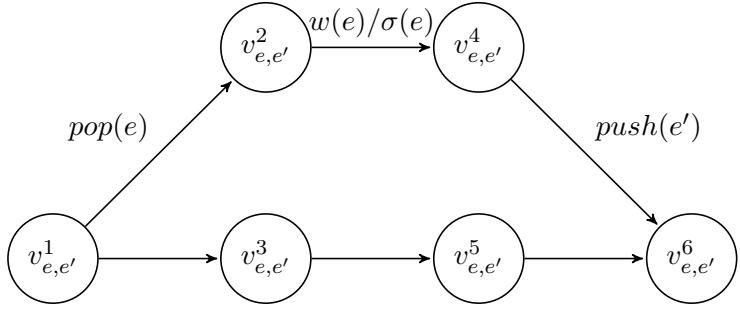
 
 Now choose an order $\le_E$ of $E$ such that for each pair $uv, vw\in E$, the edge $uv$ comes before $vw$. This order can be iteratively constructed from a topological order~$\le_V$ of~$V$: For each vertex $v$ along $\le_V$ iteratively add the edges entering $v$ to $\le_E$ as the new maximum. From $\le_E$ we construct an order $\le_G$ of the gadgets $G_{e,e'}$ by defining \[G_{e_1,e_2} \le_G G_{e_3,e_4} \leftrightarrow e_1< e_3 \vee (e_1 = e_3 \land e_2 < e_4).\] 
 
 We now connect the gadgets along the order $\le_G$ in the following way: Let $G_{e_1,e_2}$ and $G_{e_3,e_4}$ be two successors in $\le_G$. We connect $v_{e_1,e_2}^6$ to $v_{e_3,e_4}^1$ by a $\nop$-edge of weight~$1$. Let $G_{e,e'}$ be the minimum of $\le_G$. We add a new vertex $s$ and the edge $s v_{e,e'}^1$ with weigth $1$ and stack opeation $\sigma(sv_{e,e'}^1):= \push(e_s)$ where $e_s$ is the single outgoing edge of $s$ in $G$. Let now  $G_{e,e'}$ be the maximum gadget in $\le_G$. We add a new vertex $t$ and the edge $v_{e,e'}t$ with weight $1$ and stack operation $\pop(e_t)$. This concludes the construction of $G'$.
 
 It is easy to see that $G'$ has indeed width $2$. Thus we only need to show that $G$ and $G'$ compute the same polynomial. This will follow directly from the following claim:
 
 \begin{clm}
  There is a bijection $\pi$ between the stack-realizable paths in $G$ and $G'$. Furthermore $w(P):= w'(\pi(P))$ for each stack-realizable path in $G$.
 \end{clm}
\begin{proof}
Clearly every $s$-$t$-path must traverse all gadgets in $G'$. Furthermore, whenever a gadget is entered, the stack contains only one symbol from $E$ which lies at the top of the stack. Through each gadget $G_{e,e'}$ there are exactly the two paths $v_{e,e'}^1v_{e,e'}^2v_{e,e'}^4v_{e,e'}^6$ and $v_{e,e'}^1v_{e,e'}^3v_{e,e'}^5v_{e,e'}^6$. We call the former the \emph{weighted path} through $G_{e,e'}$. For a stack-realizable $s$-$t$-path $P=e_1 e_2 \ldots e_k$ through $G$ we define $\pi(P)$ to be the unique path through $G'$ that takes the weighted path through exactly the gadgets $G_{e_i, e_{i+1}}$ for $i=1, \ldots, k=1$. We have $w(P):= w'(\pi(P))$ with this definition, because only the weighted edges in the gadgets have a weight different from $1$ in $G'$. So it suffices to show that $\pi$ is indeed a bijection.

We first show that $\pi$ maps stack-realizable paths in $G$ to stack-realizable paths in $G'$. So let~$P$ be as before. Observe that~$\pi(P)$ traverses the gadgets $G_{e_i, e_{i+1}}$ in the same order as~$P$ traverses the edges~$e_i$. Furthermore, whenever $\pi(P)$ enters a gadget $G_{e_i, e_{i+1}}$ the top stack symbol is $e_i$ and the rest of the stack content is exactly that on~$P$ before traversing~$e_i$. When leaving $G_{e_i, e_{i+1}}$ the stack content is that after traversing~$e_i$ on~$P$ with an additional symbol~$e_{i+1}$ on the top. Thus all stack operations along~$\pi(P)$ must be legal and the stack is empty after traversing the last edge towards~$t$. Thus~$\pi(P)$ is indeed stack-realizable.

Clearly, $\pi$ is injective, so to complete the proof of the claim we only need to show that it is surjective. So let $P'$ be a stack-realizable $s$-$t$-path in $G'$. Let $G_{e_1, e_1'}, \ldots , G_{e_k, e_k'}$ be the gadgets in which $P'$ takes the weighted path in the order in which they are visited. We claim that $e_se_1\ldots e_k$ is a stack-realizable $s$-$t$-path. Clearly, $s$ is the first vertex of $P$. Also in $P'$ the symbol $e_t$ is popped in the last step by construction of $G'$, so the last gadget in which $P'$ took a weighted path must be one of the form $G_{e, e_t}$, because otherwise $e_t$ cannot be the top symbol on the stack before the last step. Thus $t$ is the last vertex of $P$.

To see that $P$ is a path, observe that we have $e_i' = e_{i+1}$. Otherwise $P$ cannot have the right top symbol when taking the weighted path in $G_{e_{i+1}, e_{i+1}'}$. Thus $e_{i+1}$ must be a successor of $e_i$ in $G$ and $P$ is an $s$-$t$-path.

To see that $P$ is stack-realizable observe that when $P'$ traverses the weighted edge of a gadget $G_{e_i, e_i'}$ it has the same stack content as when $P$ traverses $e_i$ in $G$. So $P$ is obviously stack-realizable because $P'$ is.

Observing that obviously $w(P)= P'$ by construction completes the proof.
\end{proof}

In a final step we now reduce the stack symbol size to $\{0,1\}$ in a straightforward way. Let $\ell:=\lceil\log(|S\cup E|)\rceil$, then each stack symbol $s$ can be encoded into a $\{0,1\}$-string $\mu(s)$ of length~$\ell$. Now we substitute each edge $e$ of $G'$ by a path $P_e$ of length~$\ell$. If $\sigma'(e)=\push(s)$ we the edges along $P_e$ are $\push$-edges, too, that push $\mu(e)$ onto the stack. If $\sigma'(e)=\pop(s)$ we pop $\mu(s)$ in reverse order along $P_e$. If $e$ is a $\nop$-edge, all edges of $P_e$ are $\nop$-edges, too. Finally, we give one of the edges in $P_e$ the weight $w'(e)$, while all other edges get weight $1$. Doing this for all edges, it is easy to see that the resulting SBP computes the same polynomial as $G'$. Furthermore, its width is $2$.
\end{proof}

\subsection{Depth reduction}

In this section we show that the characerization of \vp\;by SBPs allows us to directly use results from counting complexity that rely on NAuxPDAs. We demonstrate this by adapting a proof by Niedermeier and Rossmanith \cite{NR95} to reprove the classical parallelization theorem for \vp\;originally proved by Valiant et al. \cite{VSBR83}. While neither the result nor the proof technique is new in itself, we argue that the use of applying the techniques using SBPs results in a proof that is arguably more transparent than any other proof of this classical theorem that we know. This raises our hopes that the SBP characterization of \vp\; may be helpful in the future.

We now start presenting the ideas of Niedermeier and Rossmanith in detail. The basic idea is the following: The realizable paths are recursively cut into subpaths and the polynomials are then computed by combining the polynomials of the subpaths. In order to reach logarithmic depth we have to make sure that the paths are cut in paths of approximately equal length to result in a balanced computation. This is complicated by that fact that the paths have to be realizable, so we have to account for the content of the stack during the computation.

We now give the details of the construction. Let $G$ be an SBP and let $P$ be a realizable path in $G$ from $a$ to $b$. Let $c$ be a vertex on $P$, then the stack height of $P$ in $c$ is the number of $\push$-edges minus the number of $\pop$-edges on $P$ from $a$ up to $c$. Similarly to Niedermeier and Rossmanith we give to a path $P$ a description $(a,b,i)$, where $a$ is the start vertex, $b$ the end vertex and $i$ the length of $P$.

We define a relation $\vdash$ on paths in order to decompose them. Let $P$ be a path with realizable subpaths $P_1$ and $P_2$ and let these three paths have the descriptions $(a,b,i)$, $(c,d,j)$ and $(e,b,k)$. Then we write $P_1, P_2 \vdash P$ iff
\begin{itemize}
 \item the stack height of $P$ on $e$ is $0$
 \item there is an $s\in S$ such that $\sigma(ac)= \push(s)$ and $\sigma(de) = \pop(s)$ and
 \item $j+k = i-2$.
\end{itemize}

We state a Lemma by Niedermeier and Rossmanith:
\begin{lem}\label{lem:dec1}
 Let $P$ be a path with description $(a,b,i)$ and $i\ge 2$. Then there exist uniquely described subpaths $P_1$, $P_2$ and $P_3$ with descriptions $(c,d,i_1)$, $(e,f,i_2)$ and $(g,d,i_3)$ with $i_2, i_3 \le i/2 < i_1$ such that $P_2, P_3 \vdash P_1$.
\end{lem}

Lemma \ref{lem:dec1} allows us to cut a path $P$ into three parts $P_2$, $P_3$ and $P-P_1$. None of these parts is too big, but we cannot iterate this procedure, because unfortunately $P-P_1$ is not a path because it has a ``gap'' from $c$ to $d$. To remedy the situation Niedermeier and Rossmanith formalize this notion of a path with gap in the following way: A \emph{path with gap} with description $(a, (c,d, j), b, i)$ consists of two paths, one from $a$ to $c$ and one from $d$ to $b$, where $i$ and $j$ with $i\ge j$ are even natural numbers. $P$ is realizable, if identifying $c$ and $d$ results in a realizable path of length $i-j$. Observe that $P$ with description $(a, (c,d,i), b,i)$ is realizable if and only if $a=c$ and $b=d$, i.e. the path consists only of a gap.

We now extend the relation $\vdash$ to paths with gaps. This is complicated a little by the fact that the gap can lie in either of the two subpaths that we want to split a path with gap into. So let $P$ be a path with gap and description $(a, (c,d,j), b,i)$. For the first case let $P_1$ be a subpath with gap and description $(e,(c,d,j),f,k)$ and $P_2$ be a subpath with description $(g,b,l)$. For the second case let $P_1$ be a subpath with description $(e,f,k)$ and $P_2$ a subpath with gap and description $(g,(c,d,j),b,l)$. Then we write $P_1, P_2 \vdash P$ if and only if the stack height $g$ is $0$, there is an $s\in S$ such that $\sigma(ac)= \push(s)$ and $\sigma(de) = \pop(s)$ and $k+l = i-2$. Observe that if $c=d$ and $j=0$ this definition coincides with the definition on paths without gap.

Niedermeier and Rossmanith give a version of Lemma \ref{lem:dec1} for paths with gap.
\begin{lem}\label{lem:dec2}
 Let $(a, (c,d,j), b,i)$ with $i-j \ge 2$ be a realizable path with gap. Then there exist uniquely determined paths $P_1$, $P_2$ and $P_3$ such that $P_1$ has the description $(e, (c,d,j), f, i_1)$, $P_2, P_3 \vdash P_1$ and either 
\begin{enumerate}
 \item $P_2$ has the description $(g,(c,d,j), h, i_2)$ and $P_3$ has the description $(k,f, i_3)$ such that $i_2-j \le (i-j)/2 < i_1 -j$ or 
\item $P_2$ has the description $(g, h, i_2)$ and $P_3$ has the description $(k,(c,d,j), f, i_3)$ such that $i_3-j \le (i-j)/2 < i_1 -j$. 
\end{enumerate}
\end{lem}

Let $P$ be a realizable path with gap with description $(a,(c,d, j), b,i)$. Then we define its weight $w(P):= w(P')$ where $P'$ is the realizable path we get from $P$ when we identify $c$ and $d$. Let $w(a,b,i):= \sum_P w(P)$ where the sum is over all realizable paths with description $(a,b,i)$. Furthermore, $w(a,(c,d,j),b,i):= \sum_P w(P)$ where the sum is over all realizable paths with gap with description $(a,(c,d,j),b,i)$. With these definitions and the Lemmas \ref{lem:dec1} and \ref{lem:dec2} we get the following Lemma:

\begin{lem}\label{lem:parallel}
\begin{enumerate}
 \item[a)] \[w(a,b,i) = \sum w(a, (c,d,j),b,i) w(e,f,i_1) w(g,d,i_2) w(ce) w(fg)\]
where the sum is over all $c,d,e,f,g\in V(G)$ such that there is an $s$ with $\sigma(ce) = \push(s)$ and $\sigma(fg)=\pop(s)$ and all even numbers $j, i_1, i_2$ with $i_1, i_2 \le i/2 < j$ and $i_1 + i_2 = j-2$.
\item[b)] \begin{align*}
 w(a,(c,d,j),b,i) \\=&\sum w(a, (c_1, d_1, j_1), b, i) w(e, (c,d,j), f, i_1) w(g,d_1, i_2) w(c_1e) w(fg)\\
&+\sum w(a, (c_1, d_1, j_1), b, i) w(e, f, i_1) w(g,(c,d,j), d_1, i_2)w(c_1e) w(fg)
\end{align*}
where both sums are over all $c_1,d_1,e, f, g\in V(G)$ such that there is an $s$ with $\sigma(c_1e) = \push(s)$ and $\sigma(fg)=\pop(s)$. The first sum is also over all even numbers $j_1, i_1, i_2$ with $i_1-j \le (i-j)/2 < j_1-j$ and $i_1 + i_2 = j_1-2$, while the second sum is over all even numbers $j_1, i_1, i_2$ with $i_2-j \le (i-j)/2 < j_1-j$ and $i_1 + i_2 = j_1-2$.
\end{enumerate}
\end{lem}
\begin{proof}
 (Sketch) For a) oberve that the decomposition of Lemma \ref{lem:dec1} is unique. So we sum the weight of every path from $a$ to $b$ of length $i$ exactly once. For b) use Lemma \ref{lem:dec2} for the same argument.
\end{proof}

The following Lemma is now easy to see:
\begin{lem}
 Let $G$ be an SBP. Then $f_G$ can be computed by a semi-unbounded circuit of depth $O(\log(|G|))$ and size $|G|^{O(1)}$.
\end{lem}

Combined with Theorem \ref{thm:charvp} we get the parallelization Lemma by Valiant et al. \cite{VSBR83}.

\begin{cor}
 Let $(f_n) \in \vp$. Then $(f_n)$ can be computed by a family of semi-unbounded circuits of polynomial size and logarithmic depth in $n$.
\end{cor}

\section{Random access memory}

\subsection{Definition}

We change the model of computation by allowing random access memory instead of a stack. We still work over a symbol set $S$ like for SBPs but we introduce three \emph{random access memory operations}: The operation $\writ$ and $\delete$ take an argument $s\in S$ while the operation $\nop$ again takes no argument. Let $\op(s)$ be a random access memory operation with $\op \in \{\writ, \delete \}$ and $P=\op_1 \op_2 \ldots \op_r$ a sequence of memory operations. With $\mathop{occ}(P, \op(s))$ we denote the number of occurences of $\op(s)$ in $P$. We call a sequence $P$ realizable if for all symbols $s\in S$ we have that $\mathop{occ}(P, \writ(s)) = \mathop{occ}(P, \delete(s))$ and for all prefixes $P'$ of $P$ we have $\mathop{occ}(P', \writ(s)) \ge \mathop{occ}(P', \delete(s))$ for all $s\in S$.

Intuitively the random access memory operations do the following: $\writ(s)$ writes the symbol $s$ into the random access memory. If $s$ is already there it adds it another time. $\delete(s)$ deletes one occurence of the symbol $s$ from the memory if there is one. Otherwise an error occurs. $\nop$ is the ``no operation'' operation again like for SBPs. A sequence of operations is realizable if no error occurs during the deletions and starting from empty memory the memory is empty again after the sequence of operations.

\begin{defi}
A \emph{random access branching program} (RABP) $G$ is an ABP with an additional edge labeling $\sigma: E \rightarrow \{ \op(s) \mid \mathop{op} \in \{\writ , \delete\}, s \in S\} \cup \{\nop\}$. A path $P = v_1 v_2 \ldots v_r$ in $G$ has the sequence of random access memory operations $\sigma(P) := \sigma(v_1v_2)\sigma(v_2v_3) \ldots \sigma(v_{r-1}v_r)$. If $\sigma(P)$ is realizable we call $P$ a \emph{random-access-realizable path}. The RABP $G$ computes the polynomial \[f_G = \sum_{P} w(P),\] where the sum is over all random-access-realizable $s$-$t$-paths $P$.
\end{defi}

In a completely analogous way to Proposition \ref{prop:nonopSBP} we can proof that disallowing $\nop$-edges does not change the computational power of RABPs.

\begin{prop}
 Let $G$ be an RABP of size $s$. There is an SBP $G'$ of size $O(s^2)$ such that $f_G= f_{G'}$ and $G'$ does not contain any $\nop$-edges. If $G$ is layered with width $k$, then $G'$ is layered, too, and has width at most $k^2$.
\end{prop}
% The remarks we made on $\nop$-edges in SBPs hold for RABPs, too. In particular we can assume that if we are not talking about bounded width RABPs $\nop$-edges do not occur, while we use them in the constructions of RABPs.

\subsection{Characterizing VNP}

Intuitively random access on the memory allows us more fine-grained control over the paths in the branching program that contribute to the computation. While in SBPs nearly all of the memory content is hidden, in RABPs we have access to the complete memory at all times. This makes RABPs more expressive than SBPs which is formalized in the following theorem.

\begin{thm}\label{thm:RABP}
 $(f_n) \in \vnp$ if and only if there is a family of polynomial size RABPs computing $(f_n)$.
\end{thm}

Again we prove the theorem in two independent lemmas, starting with the upper bound which is very easy.

\begin{lem}\label{lem:containvnp}
 If $(f_n)$ is computed by a family of polynomial size RABPs, then $(f_n) \in \vnp$.
\end{lem}
\begin{proof}
 This is easy to see with Valiant's criterion (Lemma \ref{lem:criterion}) and the fact that checking if a path through a RABP is realizable is certainly in $\mathsf{P}$.
\end{proof}

We will now show the lower bound of Theorem \ref{thm:RABP}. We will prove it directly for bounded width RABPs. To do so we consider the following dominating-set polynomial for a graph $G =(V,E)$:

\[DSP_G(X_1, \ldots , X_n) := \sum_{D} \prod_{v\in D}X_v,\] where the sum is over all dominating sets $D$ in $G$.

In Appendix \ref{app:dshard} we show that the is a family $(G_n)$ of graphs such that the resulting family $(DSP_{G_n})$ of polynomials is $\vnp$-complete.

\begin{lem}
For each family $(f_n)\in \vnp$ there is a family of width $2$ RABPs of polynomial size computing $(f_n)$.%\todo{OK, should I just not talk about fields? should I work in the constant free model?}
\end{lem}
\begin{proof}
 We will show that for a graph $G=(V,E)$ with $n$ vertices there is a RABP of size $n^{O(1)}$ and width $2$ that computes $DSP_G(X_1, \ldots , X_n)$. The RABP works in two stages. The symbol set of the RABP will be $V$. In a first stage it iteratively selects vertices $v$ and writes $v$ and all of its neightbors into the memory. In a second stage it checks that each vertex $v$ was written at least once into the memory, i.e., either $v$ or one of its neighbors was chosen in the first phase. Thus the set of chosen vertices must have been a dominating set.
 
 So fix a graph $G$. For each vertex $v$ with neighbors $v_1, \ldots , v_k$ we construct a gadget $G_v$ as shown in Figure \ref{fig:chooseDS}. We call the path through $G_v$ with the edges that have memory operations the \emph{choosing path}. Now for each vertex $v$ we construct a second gadget $G'_v$ that is shown in Figure \ref{fig:checkDS}. Choose an order on the vertices. For each non-maximal vertex $v$ in the order with successor $u$, we connect the sink of $G_v$ to the source of $G_u$ and the sink of $G_v'$ to the source of $G_u'$ with a $\nop$-edge of weight $1$. Finally, let $x$ be the maximal vertex in the order and $y$ the minimal vertex. Connect the sink of $G_{x}$ to the source of $G_{y}'$ again by a $\nop$-edge of weight $1$.

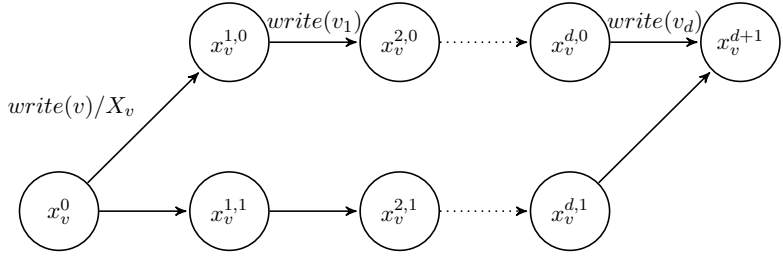
\begin{figure}[t]
 \centering
 \begin{tikzpicture}[->,>=stealth',shorten >=1pt,auto,node distance=2.8cm,semithick,scale=0.8, every node/.style={transform shape}]
  \tikzstyle{every state}=[minimum size=13mm]

  \node[state]         (A)  {$x_{v}^0$};
  \node[state]         (D) [right of=A] {$x_{v}^{1,1}$};
  \node[state]         (B) [above of=D] {$x_{v}^{1,0}$};
  \node[state]         (C) [right of=B] {$x_{v}^{2,0}$};
  \node[state]         (E) [right of=D]       {$x_{v}^{2,1}$};
  \node[state]         (H) [right of=C] {$x_{v}^{d,0}$};
  \node[state]         (J) [right of=E]       {$x_{v}^{d,1}$};
  \node[state]         (F) [right of=H]       {$x_{v}^{d+1}$};

  \path (A) edge              node {$\writ(v)/X_v$} (B)
            edge              node {} (D)
        (B) edge              node {$\writ(v_1)$} (C)
        (D) edge              node {} (E)
        (H) edge              node {$\writ(v_d)$} (F)
        (J) edge              node {} (F);

  \path[dotted] (E) edge              node {} (J)
        (C) edge              node {} (H);

\end{tikzpicture}
\caption{The gadget $G_v$. Let $v$ be a vertex with neighbors $v_1, \ldots, v_d$. The weight of $x_v^0 x_v^{1,0}$ is $X_v$ while all other edges have weight $1$. $G_v$ has two paths. Every realizable path that traverses $G_v$ on the upper path writes $v$ and all of its neightbors into the memory. This path has weight $X_v$. Realizable paths through the upper path do not change the memory in $G_v$ and have a weight weight contribution of $1$ in $G_v$.}\label{fig:chooseDS}
\end{figure}

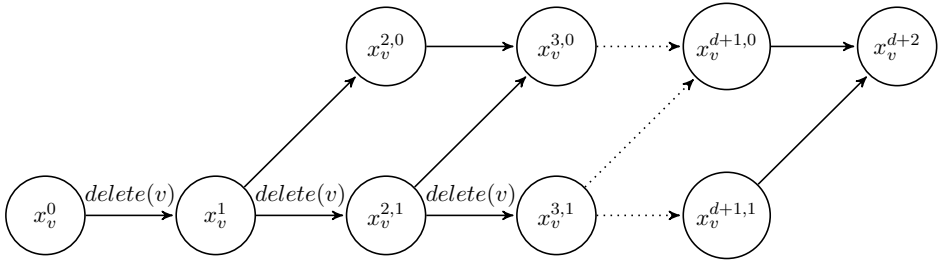
\begin{figure}[t]
 \centering
\begin{tikzpicture}[->,>=stealth',shorten >=1pt,auto,node distance=2.8cm,semithick,scale=0.8, every node/.style={transform shape}]
  \tikzstyle{every state}=[minimum size=13mm]

  \node[state]         (X)              {$x_v^0$};
  \node[state]         (A) [right of=X] {$x_{v}^1$};
  \node[state]         (D) [right of=A] {$x_{v}^{2,1}$};
  \node[state]         (B) [above of=D] {$x_{v}^{2,0}$};
  \node[state]         (C) [right of=B] {$x_{v}^{3,0}$};
  \node[state]         (E) [right of=D]       {$x_{v}^{3,1}$};
  \node[state]         (H) [right of=C] {$x_{v}^{d+1,0}$};
  \node[state]         (J) [right of=E]       {$x_{v}^{d+1,1}$};
  \node[state]         (F) [right of=H]       {$x_{v}^{d+2}$};

  \path (X) edge              node {$\delete(v)$} (A)
        (A) edge              node {} (B)
            edge              node {$\delete(v)$} (D)
        (B) edge              node {} (C)
        (D) edge              node {$\delete(v)$} (E)
            edge              node {} (C)
        (H) edge              node {} (F)
        (J) edge              node {} (F);

  \path[dotted] (E) edge              node {} (J)
        (C) edge              node {} (H)
        (E) edge              node {} (H);

\end{tikzpicture}
 \caption{The gadget $G_v'$. Let $d$ be the degree of $v$, then $G_v'$ has $d+3$ layers. All edges have weight $1$. The edges connecting vertices in the lower level have operation $\delete(v)$ while all other edges have no memory operation. Every realizable path through $G_v'$ has weight $1$ and deletes between $1$ and $d+1$ occurences of the symbol $v$ from memory.}\label{fig:checkDS}
\end{figure}

 We claim that $G'$ computes $DSP_G$. To see this, define the weight of a vertex set $D$ in $G$ to be $w(S) := \prod_{v\in S} X_v$. The following claim completes the proof.
 
 \begin{clm}
  There is a bijection $\pi$ between dominating sets in $G$ and RA-realizing paths in $G'$ such that for each dominating set $D$ in $G$ we have $w(D):= w(\pi(D))$.
 \end{clm}
\begin{proof}
Observe that for RA-realizing paths through $G'$ once the path through the gadgets $G_v$ is chosen, then rest of the path is fixed. So each RA-realizing path $P$ can be described completely by the $v$ for which the choosing paths through $G_v$ is taken.

Let $D$ be a dominating set. Let $\mathcal{P}$ be the set of $s$-$t$-paths in $G'$ that for each $v\in D$ take the choosing path through $G_v$ and for each $G'$ take the other path. Because $D$ is dominating, after a path $P\in \mathcal{P}$ has passed through the gadgets $G_v$, it contains each symbol $v\in V$ at least once. Thus there is a unique path in $\mathcal{P}$ that is RA-realizing. Call this path $\pi(D)$.

Obviously, $\pi$ is injective. To show that it is surjective, too, consider an RA-realizable path $P$ in $G'$. Let $D$ be the set of $V\in V$ for which $P$ takes the choosing path. The path $P$ passes every gadget $G_v'$, so each element $v\in V$ gets deleted from the memory at least once. It follows that each $v\in V$ must have been written to memory at least once before. So for $v\in V$ the path $P$ must go through $G_v$ or through $G_u$ for a neighbor $u$ of $v$. It follows that $D$ is a dominating set. Furthermore, $\pi(D) =P$, so $\pi$ is surjective.

Finally, $w(D):= w(\pi(D))$ is true, because the only weighted edges in $G'$ are in the gadgets $G_v$ and for each $v$ the weighted edge in $G_v$ has the weight $X_v$.
\end{proof}
 
Observing that $G'$ has width $2$, completes the proof.
\end{proof}

% \section{Open problems}
% 
% Can we use all this to widen the chasm at depth four? Can we get rid of the factor $\log(d)$ in the exponent there?
% 
% What happens if we restrict the width of the SBPs AND the symbol alphabet size?
% 
% It is easy to see that with two stacks one can simulate random access and thus one gets a characterization of \vnp. If we restrict the access to the stacks like in \cite{LM09}, do we get \vp\; again?
% 
% What is the computational power of bounded width SBPs with one stack symbol?

\paragraph*{Acknowledgements:} The author would like to thank Sébastien Tavenas for pointing out an error in an earlier proof of Lemma \ref{lem:charvpseconddirection}. The corrected proof presented in this paper is the result of discussions with him and Pascal Koiran. The author is very thankful for this contribution. Furthermore, the author is grateful to Guillaume Malod who gave very helpful feedback on a draft of this paper. Finally, the author would like to thank Peter Bürgisser and Meena Mahajan for encouraging him to write up these results as a paper.

\bibliographystyle{plain}
% \bibliography{realizable}

\begin{appendix}
 \section{VNP-completeness of the dominating-set polynomial}\label{app:dshard}
 
 In this appendix we show that there is a family of graphs such that the the polynomial family $(DSP_{G_n} := \sum_{D} \prod_{v\in D}X_v)$ is $\vnp$-complete. With Valiant's criterion (Lemma \ref{lem:criterion}) containment in $\vnp$ is clear. 
 
 For hardness we will reduce from the polynomial $VCP_G = \sum_{S} \prod_{v\in S}X_v$ where the sum is over all vertex covers $S$ of $G$. This polynomial was introduced by Briquel and Koiran \cite{BriquelK2009} who showed the following hardness result:
 
 \begin{lem}[\cite{BriquelK2009}] \label{lem:VChard}There exists a family $G_n$ of polynomial size graphs such that $(VCP_{G_n})$ is $\vnp$-complete.
  \end{lem}

With Lemma \ref{lem:VChard} it suffices to show that for every graph $G$ there is a graph $G'$ of size polynomial in the size of $G$ such that $$VCP_G \le DSP_{G'}.$$ So let $G=(V,E)$ be a graph. We construct $G'$ by adding for each $e=uv\in e$ a vertex $v_e$ and the edges $uv_e$ and $vv_e$. Every dominating set $D$ of $G'$ must contain $v_e$ or one of $u,v$. Thus $D$ is either a vertex cover of $G$ or it contains a vertex $v_e$ for an $e\in E$. Setting $X_{v_e}:=0$ one gets $VCP_G$ as the projection of $DSP_{G'}$. This finishes the proof.
 
\end{appendix}

\end{document}